%% file: main.tex
\documentclass[11pt]{article}
\usepackage{amsmath}
\usepackage{amsthm}
\usepackage{amssymb}
\usepackage{algorithm}
\usepackage{subfig}
\usepackage{color}
\usepackage[english]{babel}
\usepackage{graphicx}
\usepackage{wrapfig,epsfig}
\usepackage{epstopdf}
\usepackage{url}
\usepackage{graphicx}
\usepackage{color}
\usepackage{epstopdf}
\usepackage{algpseudocode}
\usepackage{scrextend}
\usepackage[T1]{fontenc}
\usepackage{bbm}
\usepackage{comment}
\usepackage{complexity}

\usepackage{tikz}
\usepackage{hyperref}  
\hypersetup{colorlinks=true,citecolor=blue,linkcolor=blue} 
\usetikzlibrary{arrows}
\usepackage[margin=1in]{geometry}
\graphicspath{{./figs/}}

\author{
  Aviad Rubinstein\thanks{\texttt{aviad@cs.stanford.edu}. Stanford University. Most of the work done while being a Rabin Postdoc at Harvard University.}\\
  \and
  Zhao Song\thanks{\texttt{magic.linuxkde@gmail.com}. Simons at Berkeley. Most of the work done while visiting Harvard University and hosted by Jelani Nelson.}\\
}

\date{}
\title{Reducing approximate Longest Common Subsequence to approximate Edit Distance\thanks{The authors would like to thank Alexandr Andoni, Arturs Backurs, Ilya Razenshteyn, Saeed Seddighin, Erik Waingarten for encouraging us to release this paper. The authors would also like to thank Lijie Chen and Rasmus Kyng for useful discussions. The authors would like to thank for Lijie Chen and Zhengyu Wang for proof-reading.}}

\newtheorem{theorem}{Theorem}[section]
\newtheorem*{theorem*}{Theorem}
\newtheorem{lemma}[theorem]{Lemma}
\newtheorem{definition}[theorem]{Definition}

\newtheorem{corollary}[theorem]{Corollary}
\newtheorem*{corollary*}{Corollary}

\newtheorem{fact}[theorem]{Fact}

\theoremstyle{remark}

\newtheorem*{remark*}{Remark}

\newcommand{\wh}{\widehat}
\newcommand{\wt}{\widetilde}

\renewcommand{\varepsilon}{\epsilon}
\renewcommand{\tilde}{\wt}

\DeclareMathOperator{\LCS}{LCS}

\DeclareMathOperator{\ED}{ED}

\newcommand{\ignore}[1]{}

\newcommand{\Greedy}{\textsc{Greedy}}
\newcommand{\Match}{\textsc{Match}}
\newcommand{\BMatch}{\textsc{BestMatch}}
\newcommand{\Approx}{\textsc{ApproxED}}
\newcommand{\hRB}{\widehat{R_B}}
\newcommand{\hLB}{\widehat{L_B}}
\newcommand{\hRA}{\widehat{R_A}}
\newcommand{\hLA}{\widehat{L_A}}

\makeatletter
\newcommand*{\RN}[1]{\expandafter\@slowromancap\romannumeral #1@}
\makeatother

\begin{document}

\begin{titlepage}
  \maketitle
  \begin{abstract}
\input{abstract}

  \end{abstract}
  \thispagestyle{empty}
\end{titlepage}

\newpage

\input{intro}

\input{alg}

\bibliographystyle{alpha}
\bibliography{ref}

\end{document}

%% file: abstract.tex
Given a pair of strings, the problems of computing their Longest Common Subsequence and Edit Distance have been extensively studied for decades. 
For exact algorithms, LCS and Edit Distance (with character insertions and deletions) are equivalent; the state of the art running time is (almost) quadratic and this is tight under plausible fine-grained complexity assumptions. 
But for approximation algorithms the picture is different: there is a long line of works with improved approximation factors for Edit Distance, but for LCS (with binary strings) only a trivial $1/2$-approximation was known.
In this work we give a reduction from approximate LCS to approximate Edit Distance, yielding the first efficient $(1/2+\epsilon)$-approximation algorithm for LCS for some constant $\epsilon>0$.

%% file: intro.tex

\section{Introduction}

In this paper we consider two of the most ubiquitous measures of similarity between a pair of strings: the longest common subsequence (LCS) and the edit distance.  
The LCS of two strings $A$ and $B$ is simply their longest (not necessarily contiguous) common substring. 
Edit distance is the minimum number of character insertions, deletions, and substitutions required to transform $A$ to $B$. 
In fact, under a slightly more restricted definition that does not allow substitutions%
\footnote{Since the definitions are equivalent up to a factor of $2$ (each substitution is an insertion and a deletion), this difference is irrelevant as we consider constant factor approximations of edit distance.},
the two measures are complements and the problems of computing them exactly are equivalent.

There is a textbook dynamic programming algorithm for computing LCS (or edit distance) than runs in $O(n^2)$ time, and a slightly faster $O(n^2/\log^2(n))$-time algorithm due to Masek and Paterson~\cite{mp80}. 
Finding faster algorithms is a central and long standing open problem both in theory and in practice (e.g.~Problem 35 of~\cite{k72}).
Under plausible fine-grained complexity assumptions such as SETH, neither problem can be computed much faster~\cite{AWW14-LCS, abw15, BI15, BK15, ahww16}.

For (multiplicative) approximation, the two problems are no longer equivalent. For edit distance, there is a long sequence of approximation algorithms with improving factors~\cite{bjkk04, bes06, Andoni2012, ako10, beghs18}; in particular,~\cite{cdgks18} gives a constant factor approximation in truly sub-quadratic time. For LCS with alphabet size $|\Sigma|$, in contrast, there is a trivial $1/|\Sigma|$-approximation, and no better algorithms are known (for large alphabet there are some hardness of approximation results~\cite{ab17, ar18, cglrr18} and also approximation algorithms with non-trivial polynomial factors~\cite{hsss19, RSSS19}).

In this paper we focus on binary strings, where the trivial algorithm gives a $(1/|\Sigma|=1/2)$-approximation.
Breaking this $1/2$ barrier is a well-known open problem in this area.
Our main result is a fine-grained reduction that implies obtaining a $1/2+\epsilon$-approximation for binary LCS (for some constant $\epsilon>0$) is no harder than approximating edit distance to within some constant factor.
\begin{theorem}[Reduction: approximate $\ED$ implies approximate $\LCS$]\label{thm:reduction}\hfill

Suppose that there exists a constant $c$ and an approximate edit distance algorithm that runs in time $T(n)$ and, given two binary strings $A,B$ of length $n$, returns an estimate $\tilde{\ED}(A,B) \in [\ED(A,B), c\cdot\ED(A,B) + o(n)]$. Then there exists a fixed constant $\epsilon = \epsilon(c) \in (0,1/2)$ and a deterministic approximation algorithm for longest common subsequence that runs in  deterministic $T(n)+O(n)$ and approximates $\LCS(A,B)$ to within a $(1/2+\epsilon)$-approximation factor.
\end{theorem}

\begin{remark*}
We state the above theorem in terms of estimating the edit distance or length of the LCS. If the edit distance can efficiently compute the transformation (this assumption is almost wlog by~\cite{CGKK18}), then our algorithm can also efficiently compute the common string. 
\end{remark*}

As mentioned above, the recent breakthrough of~\cite{cdgks18} gives a constant factor approximation of edit distance in truly-subquadratic ($\tilde{O}(n^{2-2/7})$) time. By plugging their algorithm into our reduction, we would obtain $(1/2+\epsilon)$-approximation algorithm for binary LCS with the same running time. 
By applying our reduction to the even more recent 
 approximation algorithms for edit distance%
\footnote{The near-linear time approximation algorithms for edit distance~\cite{KS19, BR19} also incur a sublinear additive error term, but that is OK for our reduction.} that run in near-linear time~\cite{KS19, BR19}, we obtain the following stronger corollary:

\begin{corollary}[Approximate LCS]
For every constant $\delta>0$ there exists a constant $\epsilon>0$ such that, given two binary strings $A,B \in \{0,1\}^n$, there is an algorithm that runs in $O(n^{1+\delta})$ time and 
$\LCS(A,B)$ to within a $(1/2+\epsilon)$-factor.
\end{corollary}

\subsubsection*{Technical preview}
The crux of our algorithm is analyzing first order statistics (counts of $0$s and $1$s) of the input strings ($A,B$) and their substrings.  We begin with a few simple observations. Below, we normalize $\ED$ and $\LCS$ so that they're always between $0$ and $1$ (as opposed to $0$ and $n$).
\begin{itemize}
\item 
If the strings are balanced, namely have the same number of $0$s and $1$s, we know that $\LCS(A,B) \in [1/2,1]$. If the strings are very close, say $\LCS(A,B) \geq (1-\delta)$ for sufficiently small $\delta>0$, we can use the assumed edit distance algorithm as a black box and find a common substring of length $\geq (1-O(\delta))$.
On the converse if the substring returned by the algorithm is shorter than $(1-O(\delta))$, we know that $\LCS(A,B) < 1-\delta$, and thus returning an all-$1$ string of length $1/2$ is a $(1/2+2\delta)$-approximation.
\item 
If $A$ is balanced and $B$ has e.g.~10\% $0$s and 90\% $1$s, we know that $\LCS(A,B) \in [0.5,0.6]$, so simply returning the all-$1$ string of length $1/2$ is a $5/6$-approximation.
The same holds for most ways in which one or both strings are unbalanced.
\item
However there is one difficult case when the string are {\em perfectly unbalanced}, e.g. $A$ has 99\% $0$s and $B$ has 99\% $1$s. 
Now the first order statistics over the entire strings only tell us that $\LCS(A,B) \in [0.01,0.02]$, so the trivial approximation doesn't beat $1/2$. 
On the other hand, the edit distance is at least $0.98$, so even a $1.1$-approximation algorithm for edit distance wouldn't give us a non-trivial guarantee for this case.
\end{itemize}

Our main technical contribution is a careful analysis of this last case (and its many sub-cases).

%% file: alg.tex

\section{Preliminaries}

For strings $x,y \in \{ 0, 1 \}^m$ for $m \leq n$, we use $1(x)$ to denote the number of $1$ in $x$, $0(x)$ to denote the number of $0$ in $x$,
and $\LCS(x,y)$ to denote the length of their longest common subsequence. 
All of these function are normalized w.r.t. the length of the original input to our main algorithm, $n$; in particular we always have $0(x), 1(x), \LCS(x,y) \in [0, m/n]$.

\begin{fact}\label{fact:lcs-ub}
$$\LCS(A,B) \leq \min\{0(A),0(B)\}+\min\{1(A),1(B)\}.$$
\end{fact}

\subsubsection*{Parameters $\alpha, \beta, \gamma, \delta$}
In the proof we consider the following parameters:
\begin{description}
\item[$\alpha$] We define $\alpha := \min\{1(A),1(B),0(A),0(B)\}$. Notice that $\alpha$ may be very small, and even approaching $0$ as a function of $n$.
We assume wlog that this minimum is attained by $1(A) = \alpha$.
\item[$\beta$] The parameter $\beta$ will represent a robustness parameter for some of our bounds. We take $\beta = \Theta(\alpha)$, but it may be smaller by an arbitrary constant factor.
\item[$\gamma$] The parameter $\gamma \in (0,1)$ is a constant that depends on the approximation factor $c$ of the approximation algorithm for edit distance that we assume. We choose $\beta$ sufficiently small such that $\gamma \alpha \gg \beta$.
\item[$\delta$] The parameter $\delta$ represents the deviation from ``perfectly unbalanced'' case (see Lemma~\ref{lem:unbalanced}). It is an arbitrary small constant. In particular, $\delta \alpha \ll \beta$. It is sufficiently small that for succinctness of representation we'll simply omit it (as if it were zero) after Lemma~\ref{lem:unbalanced}.
\end{description}

\subsubsection*{Subroutines}

Our reduction will assume the availability of an algorithm \Approx~which takes as input two strings $A,B$ of length $n$ and outputs $1-\tilde{\ED}(A,B)$ where $\tilde{\ED}(A,B) \in [\ED(A,B), c\cdot\ED(A,B) + o(1)]$.

In addition, we also define three trivial algorithms; they all run in time linear in length of input string.
\begin{definition}[\Match]\label{fac:match}
Given input string $A$ and $B$, and a symbol $\sigma \in \Sigma$. The algorithm $\Match(A,B,\sigma)$ will output a string $C$ where every character is $\sigma$ and the length of $C$ is $\min\{\sigma(A),\sigma(B)\}$. This algorithm takes $O(|A| + |B|)$ time.
\end{definition}

\begin{definition}[\BMatch]
Given input string $A$ and $B$. The algorithm $\BMatch(A,B)$ will take the longest one of $\Match(A,B,0)$ and $\Match(A,B,1)$. This algorithm also takes $O(|A| + |B|)$ time.
\end{definition}

\begin{definition}[\Greedy]
Given input string $A_1, A_2$ and $B$. The algorithm $\Greedy(A_1,A_2,B)$ will find the optimal contiguous partition $B=B_1\cup B_2$ so as to maximize $\BMatch(A_1,B_1) + \BMatch(A_2,B_2)$. This algorithm also takes $O(|A| + |B|)$ time.
\end{definition}

Below, we slightly abuse notation and refer to the above algorithms (\Approx, \Match, \BMatch, \Greedy) both when we want their output to be the actual common string, and the length. Which output we need will be clear from context.

\section{Reducing to perfectly unbalanced case}
In this section we formalize the intuition from the introduction that \BMatch~gives a better-than-$1/2$-approximation unless $1(A) \approx 0(B)$.

\begin{lemma}[Reduction to perfectly unbalanced case]\label{lem:unbalanced}\hfill

If $|1(A)-0(B)| > \delta\min\{0(A),0(B),1(A),1(B)\}$, then 
$$\BMatch(A,B) \geq (1/2+\delta/2)\LCS(A,B).$$
\end{lemma}
\begin{proof}
Assume wlog%
\footnote{This is wlog since $|1(A)-0(B)| = |0(A)-1(B)|$, so the premise of this lemma is symmetric.} that $1(A) = \min\{0(A),0(B),1(A),1(B)\}$. 
Then we have,
\begin{align*}
\BMatch(A,B) & = \Match(A,B,0) \\
& = \min\{0(B),0(A)\}\\
& = 0(B) && \text{(By assumption $1(A) \leq 1(B)$)}\\
& > (1+\delta)1(A)  && \text{(By premise of lemma)}\\
& = (1+\delta)\min\{1(A),1(B)\}\\
& \geq (1+\delta)\big(\LCS(A,B) - \min\{0(A),0(B)\}\big) && \text{(Fact~\ref{fact:lcs-ub})}\\
&=(1+\delta)\big(\LCS(A,B) - \BMatch(A,B)\big) .
\end{align*}
\end{proof}

We henceforth assume wlog that 
\begin{gather}\label{eq:unbalanced}
|0(A)-1(B)| \leq \delta\min\{0(A),0(B),1(A),1(B)\}.
\end{gather}

For ease of presentation, we henceforth omit $\delta$ from our calculations, i.e. we'll assume that $\delta = 0$. It will be evident that modifying any of our inequalities by factors in $[\pm \delta \alpha]$ will not affect the proofs.

Eq.~\eqref{eq:unbalanced} is very important in our analysis, but it does not rule out the perfectly balanced case, namely $0(A)\approx0(B)\approx1(A)\approx1(B)\approx1/2$.

\begin{lemma}[Ruling out the perfectly balanced case]\label{lem:balanced}\hfill
Let $\beta', \gamma > 0$ be sufficiently small constants.
If $0(A) \in [1/2 \pm \beta']$, then 
\begin{gather}\label{eq:reduction}
\max\{\BMatch(A,B),\Approx(A,B)\} \geq (1/2+\gamma)\LCS(A,B).
\end{gather}
\end{lemma}
\begin{proof}
By Eq.~\eqref{eq:unbalanced}, the premise implies that $1(B) \in [1/2 \pm \beta']$, and by symmetry also $1(A), 0(B) \in [1/2 \pm \beta']$.
Therefore, $\BMatch(A,B) \geq 1/2 - \beta'$.

Suppose that $\BMatch(A,B)$ isn't big enough to satisfy Eq.~\eqref{eq:reduction} (otherwise we're done). Then,
\begin{gather*}\LCS(A,B) > 2\BMatch(A,B) - 2\gamma \geq 2(1/2 - (\beta' + \gamma)) = 1-2(\beta'+\gamma).
\end{gather*}
Thus also 
\begin{gather*}\ED(A,B) = 1-\LCS(A,B) \leq 2(\beta' + \gamma).
\end{gather*}
Therefore, by its approximation guarantee, we have that 
\begin{gather*}\Approx(A,B) \geq 1 - c \cdot \ED(A,B)-o(1) \geq 1-2c(\beta' + \gamma)  -o(1)\geq 1/2+\gamma.
\end{gather*}
(The latter inequality follows by choosing $\beta'$ and $\gamma$ sufficiently small.)
\end{proof}

Setting $\beta'= 10\beta$, we henceforth assume wlog that 
\begin{gather}\label{eq:not-balanced}
0(A),1(A),0(B),1(B) \notin [1/2 \pm 10\beta].
\end{gather}

\section{Perfectly unbalanced strings}

In this section we build on our assumptions from Eq.~\eqref{eq:unbalanced} and Eq.~\eqref{eq:not-balanced} from the previous section to complete the proof of our reduction.

Recall that we define $\alpha:=1(A) < 1/2$, and by Eq.~\eqref{eq:unbalanced}, we also have $0(B) = \alpha$.
We partition each string into three contiguous substrings, where the extreme left and right substring are each of length $\alpha$:
\begin{align*}
A & = L_A \cup M_A \cup R_A\\
B & = L_B \cup M_B \cup R_B
\end{align*}
\begin{gather*}
|L_A| = |R_A| = |L_B| = |R_B| = \alpha
\end{gather*}
\begin{gather*}
|M_A| = |M_B| = 1-2\alpha.
\end{gather*}

\begin{figure}[!h]
\centering
	\includegraphics[width = 0.9\textwidth]{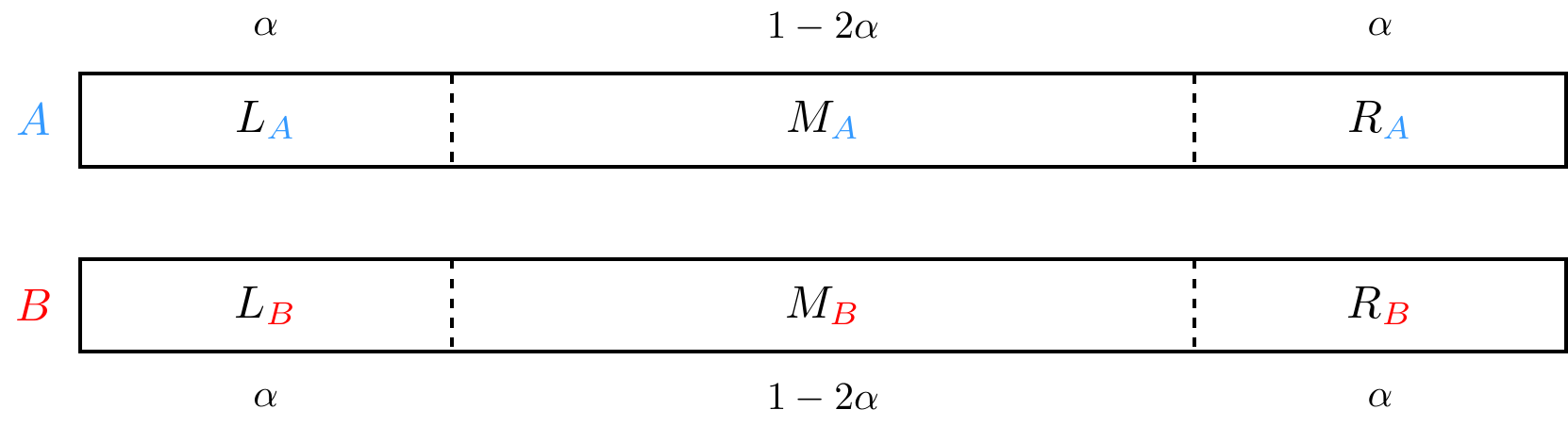}
	\caption{}\label{fig:case0}
\end{figure}

We consider six cases for the proportions of $1$'s and $0$'s in $R_A, L_A, R_B, L_B$ as in Eq.~\eqref{eq:LCS_hack_summarize_six_cases}. By Table~\ref{tab:LCS_hack_summarize_six_cases}, we know that those six cases cover all the possibilities.

Our six cases can be summarized in the following equation,
\begin{align}\label{eq:LCS_hack_summarize_six_cases}
\begin{cases}
1(R_B) \leq \alpha/2 + 2\beta , 0(R_A) \leq \alpha /2 + 2\beta & \text{~Case~1} \\
1(L_B) \leq \alpha/2 + 2\beta , 0(L_A) \leq \alpha /2 + 2\beta & \text{~Case~2}  \\
1(R_B) \leq \alpha/2 + \beta , 1(L_B) \leq \alpha /2 + \beta, 0(L_A) > \alpha /2 + 2\beta, 0(R_A) > \alpha /2 + 2\beta & \text{~Case~3}  \\
1(R_B) > \alpha/2 + 2\beta, 1(L_B) > \alpha/2 + 2\beta, 0(L_A) \leq \alpha /2 + \beta, 0(R_A) \leq \alpha /2 + \beta & \text{~Case~4}  \\
1(R_B) > \alpha/2 + \beta, 0(L_A) > \alpha/2 + \beta & \text{~Case~5}  \\
1(L_B) > \alpha/2 + \beta, 0(R_A) > \alpha/2 + \beta & \text{~Case~6} 
\end{cases}
\end{align}

\begin{table}[!h]
\begin{center}\caption{Fill all the six cases in Eq.~\eqref{eq:LCS_hack_summarize_six_cases} into the whole space. Note that $1+2+3$ means the combination of $1$, $2$ and $3$ covers it. $5,6$ means any one of them covers it.}\label{tab:LCS_hack_summarize_six_cases}
    \begin{tabular}{| p{2.9cm}| p{2.9cm} | p{2.9cm} | p{2.9cm} | p{2.9cm} |}
    \hline
     & $0(R_A)\leq \alpha/2+ \beta$, $0(L_A) \leq \alpha/2+ \beta$ & $0(R_A)\leq \alpha/2+ \beta$, $0(L_A) > \alpha/2+ \beta$ & $0(R_A) > \alpha/2+ \beta$, $0(L_A) \leq \alpha/2+ \beta$ & $0(R_A) > \alpha/2+\beta$, $0(L_A) > \alpha/2+\beta$ \\ \hline
    $1(R_B)\leq \alpha/2+ \beta$, $1(L_B) \leq \alpha/2+ \beta$  & 1,2 & 1 & 2 & 1+2+3 \\ \hline
    $1(R_B)\leq \alpha/2+\beta$, $1(L_B) > \alpha/2+ \beta$     & 1 & 1 & 6 & 6  \\ \hline
    $1(R_B) > \alpha/2+ \beta$, $1(L_B) \leq \alpha/2+\beta$    & 2 & 5 & 2 & 5  \\ \hline
    $1(R_B) > \alpha/2+\beta$, $1(L_B) > \alpha/2+\beta$       & 1+2+4 & 5 &  6 & 5,6 \\ \hline
    \end{tabular}
\end{center}
\end{table}

\subsection*{Case 1: $1(R_B) \leq \alpha /2 + 2\beta$, $0(R_A) \leq \alpha /2 + 2\beta$}

We split this case into three sub-cases, as follows:
\begin{align}\label{eq:LCS_hack_summarize_six_cases_case_1}
\begin{cases}
1(R_B) \in [\alpha/2 \pm 4\beta], ~ 0(R_A) \in [\alpha /2 \pm 4\beta ] & \text{~Case~1(a)}\\
1(R_B) < \alpha/2 - 4\beta, ~ 0(R_A)  \leq \alpha /2 + 2\beta & \text{~Case~1(b)} \\
1(R_B) \leq \alpha/2 + 2\beta, ~ 0(R_A) < \alpha /2 - 4\beta & \text{~Case~1(c)}\end{cases}
\end{align}

\subsubsection*{Case 1(a): $1(R_B) \in [ \alpha /2 \pm 4\beta ]$, $0(R_A) = [ \alpha /2 \pm 4\beta]$}

At a high level, we want to split the original problem into two subproblems: 
\begin{description}
\item[left-middle] $(L_A \cup M_A , L_B \cup M_B  )$;
\item[right] $(R_A,R_B)$. 
\end{description}
Running \BMatch~on the left-middle subproblem gives a $(1/2)$-approx; the right subproblem is (approximately) balanced so Lemma~\ref{lem:balanced} (i.e. taking the better of \BMatch~and \Approx) gives better-than-$1/2$.

The visualization of this case is presented in Figure~\ref{fig:case_1a}.

We first want to upper bound $\LCS(A,B)$ as roughly the sum of LCSs of the two subproblems, but in general this may not be the case. 
Fix an optimal matching $\mu$ corresponding to a longest common substring between $A$ and $B$. Assume wlog (by symmetry) that $\mu(R_A) \subseteq R_B$, i.e.~the LCS does not match any $R_A$ characters with characters from $L_B \cup M_B$. $\mu$ induces a new partition of $B$ into two%
\footnote{Note that we do not define a $\widehat{M_B}$.} 
contiguous substrings $\hLB \cup \hRB$ such that $\mu(R_A) \subseteq \hRB \subseteq R_B$.
By optimality of $\mu$, we have
\begin{gather}\label{eq:1a-1}\LCS(A,B) = \LCS(L_A \cup M_A, \hLB) + \LCS(R_A, \hRB).
\end{gather}
Applying Fact~\ref{fact:lcs-ub} to both terms on the RHS, we have 
\begin{align}\label{eq:1a-2}\LCS(A,B) \leq & \underbrace{\min\{1(L_A \cup M_A),1(\hLB)\} + \min\{0(L_A \cup M_A),0(\hLB)\}}_{=X} \nonumber\\
&+ \underbrace{\min\{1(R_A),1(\hRB)\}+\min\{0(R_A),0(\hRB)\}}_{=Y}.
\end{align}
We henceforth denote the left and right contributions to the bound on the LCS by $X$ and $Y$ respectively. (So $\LCS(A,B)\leq X+Y$.)
We also define:
$$ Z := \max\big\{\min\{1(L_A \cup M_A),1(\hLB)\}, \min\{0(L_A \cup M_A),0(\hLB)\}\big\}.$$
(Observe that $Z \geq X/2$.)

We now prove a lower bound on the LCS that our algorithm can find.

\begin{align}\label{eq:1a-3}
  \Greedy(L_A\cup M_A,R_A,B) 
 \geq  & \max\big\{\min\{1(L_A \cup M_A),1(\hLB)\},\min\{0(L_A \cup M_A),0(\hLB)\}\big\}  \nonumber\\
& \;\; + \max\big\{\min\{1(R_A),1(\hRB)\},\min\{0(R_A),0(\hRB)\}\big\} \nonumber\\
 \geq  & Z + Y/2.
\end{align}

We break into sub-cases, depending on the value of $Z$.

\paragraph{Case 1(a-i): $Z > \alpha/2+10\beta$}
In this case, observe that 
\begin{align*}
X - Z = & \min\{1(L_A \cup M_A),1(\hLB),0(L_A \cup M_A),0(\hLB)\} \\
\leq & 1(L_A \cup M_A) \\
= & \alpha - 1(R_A) && \text{($\alpha = 1(A)$)}\\
\leq & \alpha/2+4\beta && \text{(Case 1(a) assumption)}\\
< & Z-6\beta && \text{(Case 1(a-i) assumption)}
\end{align*}
Therefore, $Z > X/2 + 3\beta$. Combining with Eq.~\eqref{eq:1a-2} and~\eqref{eq:1a-3}, we have that
\begin{gather*}
  \Greedy(L_A\cup M_A,R_A,B)  \geq Z+Y/2 > X/2 + Y/2+3\beta \geq \LCS(A,B)/2+3\beta.
\end{gather*}

\paragraph{Case 1(a-ii): $Z \leq \alpha/2+10\beta$}

By Eq.~\eqref{eq:1a-1}, we have 
\begin{align}\label{eq:1a-4}
\LCS(A,B) \leq &  X+ \LCS(R_A, \wh{R_B}) \nonumber\\
\leq & 2Z+\LCS(R_A, \wh{R_B}) & \text{($Z \geq X / 2$)} \nonumber\\
\leq & \alpha + 20\beta +\LCS(R_A, \wh{R_B}) & \text{(Case 1(a-ii) assumption)}  \nonumber\\
\leq & \alpha + 20\beta + \LCS(R_A,R_B) . & \text{($\wh{R_B} \subseteq R_B$)}
\end{align}
For our purposes, this is effectively as good as bounding $\LCS(A,B)$ by the sum of LCSs of the left-middle and right subproblems.

We run \BMatch~on the left-middle subproblem. We have that 
\begin{align*}
\BMatch(L_A \cup M_A , L_B \cup M_B  ) 
\geq & ~ \min\{0(L_A \cup M_A),0(L_B \cup M_B)\}  \\
\geq & ~ \alpha/2- 4 \beta && \text{(Case 1(a) assumption)}.
\end{align*}

We run  \BMatch~and \Approx~on $R_A,R_B$ and take the better of the two outcomes. We apply Lemma~\ref{lem:balanced} to strings $R_A,R_B$ with $\beta'=4\beta/\alpha$. (Notice that by Case 1(a) assumption, they are guaranteed to be approximately balanced to within $\pm4\beta$, or a relative $\pm4\beta/\alpha$.) We therefore have that
\begin{align*}
\max\{\BMatch&(R_A,R_B), \Approx(R_A,R_B) \}	 \\
\geq & ~ (1/2+\gamma)\LCS(R_A,R_B) && \text{(Lemma~\ref{lem:balanced})}  \\
\geq & ~ \LCS(R_A,R_B)/2 + \gamma\alpha/2 - O(\beta \gamma)  && \text{($\LCS(R_A,R_B)\geq \alpha/2-O(\beta)$)} \\
\geq & ~ \LCS(R_A,R_B)/2 + \gamma\alpha/2 - O(\beta). && \text{($\gamma \le 1$)}
\end{align*}

So in total, our algorithm finds a common substring of length at least 
\begin{align*}
\big(\alpha + \LCS(R_A,R_B)\big)/2 +\gamma\alpha - O(\beta) 
\geq & ~ \LCS(A,B)/2 +\gamma\alpha/2 - O(\beta) && \text{(Eq.~\eqref{eq:1a-4})}\\
\geq & ~ \LCS(A,B)/2 + \frac{2}{6}\gamma\alpha &&  \text{($\gamma \alpha \gg \beta$)} \\
\geq & ~ (1/2+\gamma/6)\LCS(A,B). && \text{($\LCS(A,B) \leq 2\alpha$)}
\end{align*}

\begin{figure}[!h]
\centering
	\subfloat[Partition $(\wh{L_B},\wh{R_B})$ created by $\Greedy(L_A \cup M_A, R_A, B)$ ]{\includegraphics[width = 0.95\textwidth]{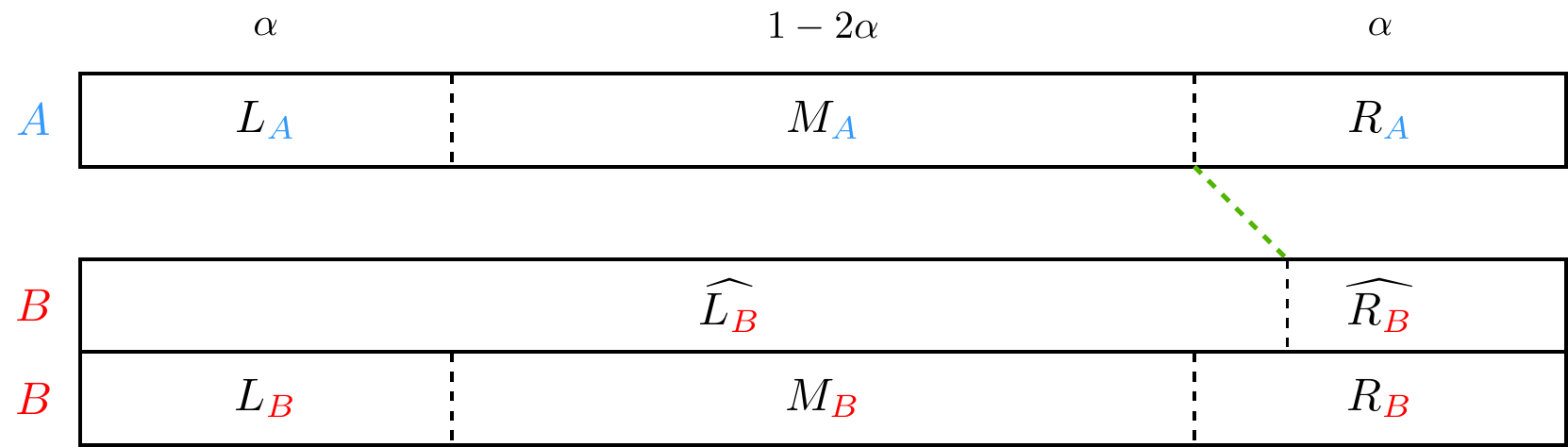}}\\
	\subfloat[$\BMatch(L_A \cup M_A, L_B \cup M_B)$ and $\max( \BMatch(R_A,R_B), \Approx(R_A,R_B) )$]{\includegraphics[width = 0.95\textwidth]{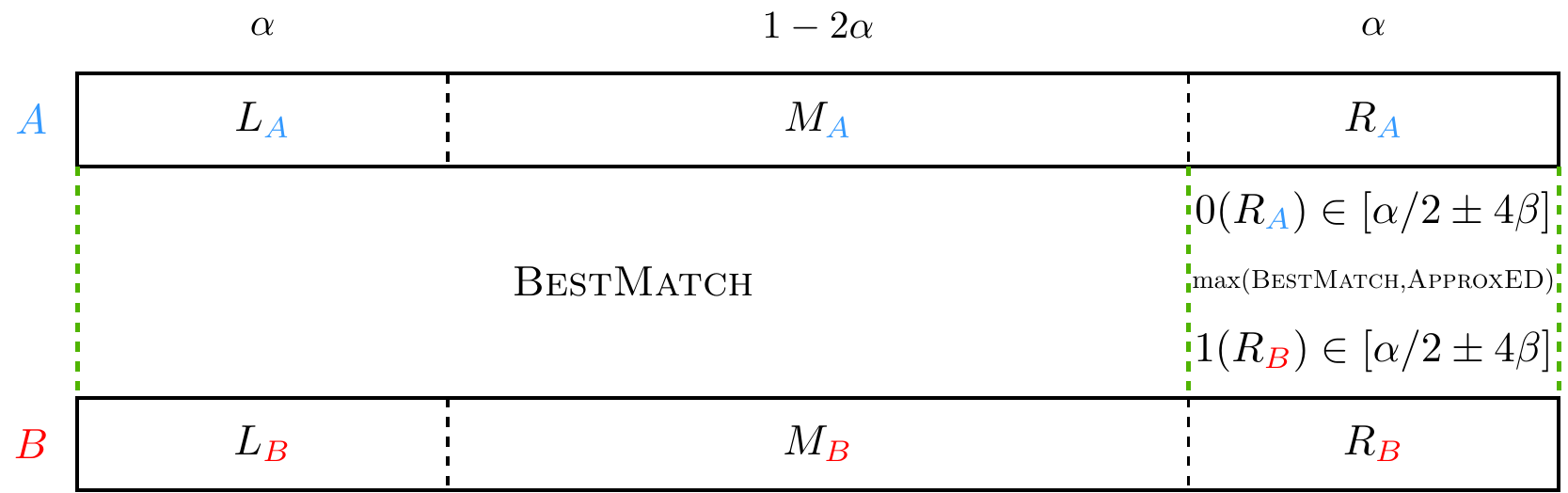}}
	\caption{Visualization of Case 1(a) which is $0(R_A) \in [\alpha/2\pm 4\beta]$ and $1(R_B) \in [\alpha/2\pm 4\beta]$. If $0(\wh{L_B}) > \alpha/2 + 10\beta$, we use $\Greedy$ result. If $0(\wh{L_B}) \leq \alpha/2 + 10\beta$, we use the result $\BMatch$+ $ \max ( \BMatch, \Approx )$.}\label{fig:case_1a}
\end{figure}

\subsubsection*{Case 1(b): $1(R_B) < \alpha /2 -4\beta$, $0(R_A) \leq \alpha/2+ 2\beta$}
Fix an optimal matching $\mu$. 
We further split this case into two sub-cases, depending on whether $\mu(R_A) \subseteq R_B$ or $R_B \subseteq \mu(R_A)$. (In Case 1(a) we could assume the former wlog by symmetry. Also notice that in general both may occur simultaneously.)

If $\mu(R_A) \subseteq R_B$, define the partition $\hLB, \hRB$ as in Case 1(a). We have
\begin{align}\label{eq:1d-i}
\LCS(A,B) & = \LCS(L_A \cup M_A, \hLB) + \LCS(R_A, \hRB) \nonumber\\
&\leq \underbrace{ 1(L_A \cup M_A) }_{ \leq \alpha/2 + 2\beta } + \underbrace{0(\hLB) + 0(\hRB)}_{\leq \alpha} + \underbrace{1(\hRB)}_{\leq \alpha/2-4\beta} && \text{(Fact~\ref{fact:lcs-ub})}\nonumber\\
& \leq 2\alpha -2\beta.
\end{align}

Similarly, if $\mu(R_A) \supseteq R_B$, we can define an analogous partition of $A$ into $\hLA, \hRA$:
\begin{align}\label{eq:1d-ii}
\LCS(A,B) & = \LCS(\hLA,L_B \cup M_B) + \LCS(\hRA, R_B) \nonumber\\
&\leq \underbrace{0(L_B \cup M_B)}_{\leq \alpha/2-4\beta} + \underbrace{1(\hLA) + 1(\hRA)}_{\leq \alpha} + \underbrace{1(R_B)}_{\leq \alpha/2-4\beta} && \text{(Fact~\ref{fact:lcs-ub})}\nonumber\\
& \leq 2\alpha -8\beta.
\end{align}

Either way, we have that $\LCS(A,B) \leq 2\alpha -2\beta$; therefore $\Match(A,B,0) = \alpha$ guarantees a better-than-$1/2$-approximation.

\subsubsection*{Case 1(c): $1(R_B) \leq \alpha /2 +2\beta$, $0(R_A) < \alpha/2-4\beta$}
Follows analogously to Case 1(b).

\subsection*{Case 2:  $1(L_B) \leq \alpha /2 + \beta $, $0(L_A) \leq \alpha /2 + \beta$} 
We reverse the order of string $A$ and $B$, then the proof is the same as Case 1.

\subsection*{Case 3: $1(R_B)\leq \alpha/2 + \beta$, $1(L_B) \leq \alpha/2 + \beta $, $0(L_A) > \alpha/2 + 2\beta$ and $0(R_A) > \alpha /2 + 2\beta$}
We visualize this case in Figure~\ref{fig:cases_3}.

We show that simple applications of \Match~to the left, middle, and right substrings can guarantee a common string of at least $\alpha+2\beta \geq \LCS(A,B)/2+2\beta$.

For the middle substrings, observe that $1(M_A) = 1(A) - 1(R_A) - 1(L_A) > 4\beta$. Also by Eq.~\eqref{eq:not-balanced}, $1(M_B) \geq 8\beta$. Therefore,
\begin{gather}\label{eq:3-m}
\Match(M_A,M_B,1) = \min\{1(M_A), 1(M_B)\} \geq 4\beta.
\end{gather}
For the left substrings, observe that $0(L_B) = |L_B| - 1(L_B) > \alpha/2-\beta$. Therefore, 
\begin{gather}\label{eq:3-l}
\Match(L_A,L_B,0) = \min\{0(L_A), 0(L_B)\} \geq \alpha/2-\beta.
\end{gather}
Similarly, 
\begin{gather}\label{eq:3-r}
\Match(R_A,R_B,0) = \min\{0(R_A), 0(R_B)\} \geq \alpha/2-\beta.
\end{gather}

Summing up Eq.~\eqref{eq:3-m},\eqref{eq:3-l},\eqref{eq:3-r}, our algorithm obtains a common string of length at least $\alpha+2\beta$.

\subsection*{Case 4: $1(R_B) > \alpha/2 + 2\beta$, $1(L_B) > \alpha/2 + 2\beta $, $0(L_A) \leq \alpha/2 + \beta$ and $0(R_A) \leq \alpha /2 + \beta$}

We visualize this case in Figure~\ref{fig:cases_4}.

If we switch $A$ and $B$, then the proof is the same as Case 3.

\subsubsection*{Case 5: $1(R_B) > \alpha/2 + \beta$, and $0(L_A) > \alpha/2 + \beta$}
We visualize this case in Figure~\ref{fig:cases_5}.

We apply \Match~to two subproblems to obtain a common substring of length greater than $\alpha+2\beta \geq \LCS(A,B)/2+2\beta$.

Observe that $0(L_B\cup M_B) = 0(B) -0(R_B) > \alpha/2+\beta$.  
\begin{gather*}
\Match(L_A, L_B\cup M_B, 0) = \min\{0(L_A),0(L_B\cup M_B)\} > \alpha_/2+\beta.
\end{gather*}
By an analogous argument, 
\begin{gather*}
\Match(M_A \cup R_A, R_B, 1) = \min\{1(M_A \cup R_A), 1(R_B)\} > \alpha/2+\beta.
\end{gather*}

\subsubsection*{Case 6: $1(L_B) > \alpha/2 + \beta$, and $0(R_A) > \alpha/2 + \beta$}
We visualize this case in Figure~\ref{fig:cases_6}.

We reverse the oder of string $A$ and $B$, then the proof is the same as Case 5.

\begin{figure}[th]
\centering
	\subfloat[]{\includegraphics[width = 0.9\textwidth]{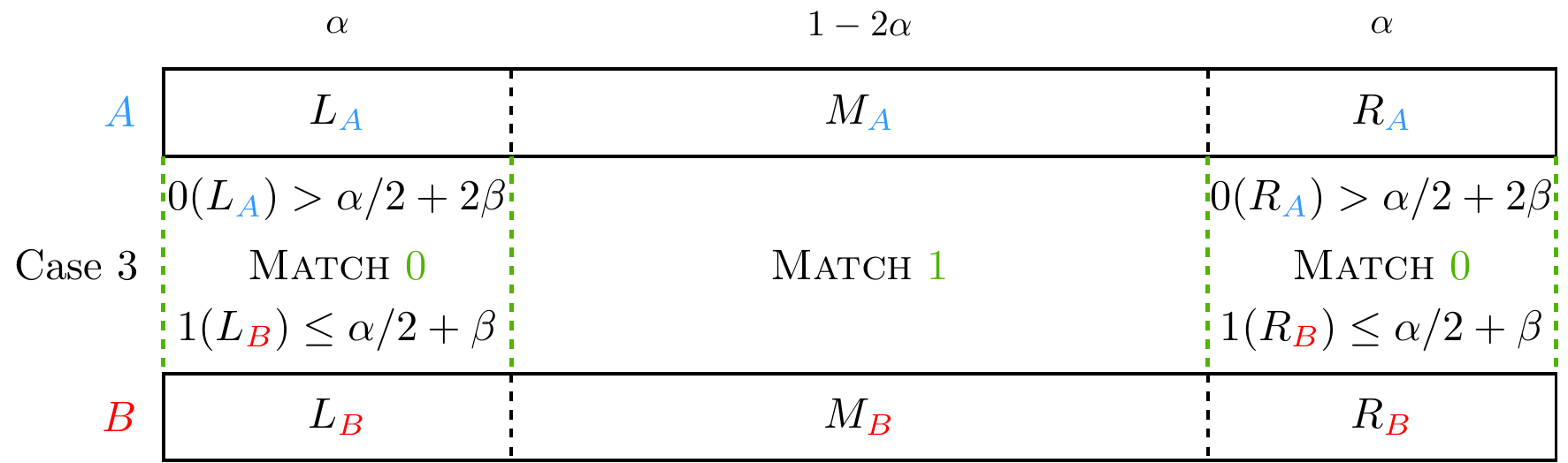}\label{fig:cases_3}}\\
	\vspace{1mm}
	\subfloat[]{\includegraphics[width = 0.9\textwidth]{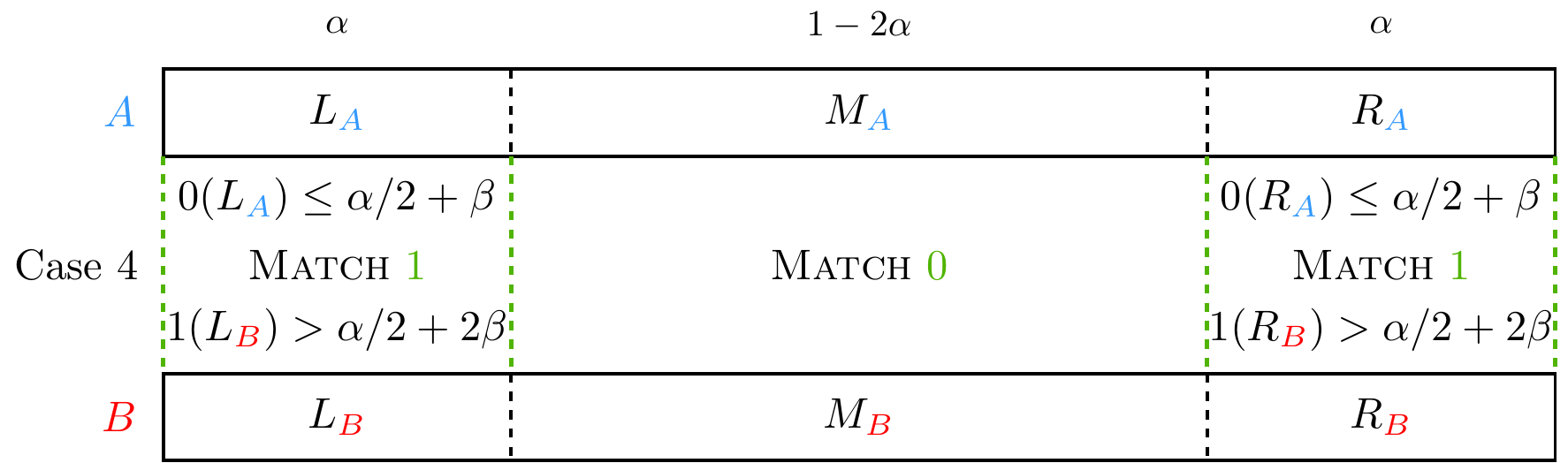}\label{fig:cases_4}}\\
	\vspace{1mm}
	\subfloat[]{\includegraphics[width = 0.9\textwidth]{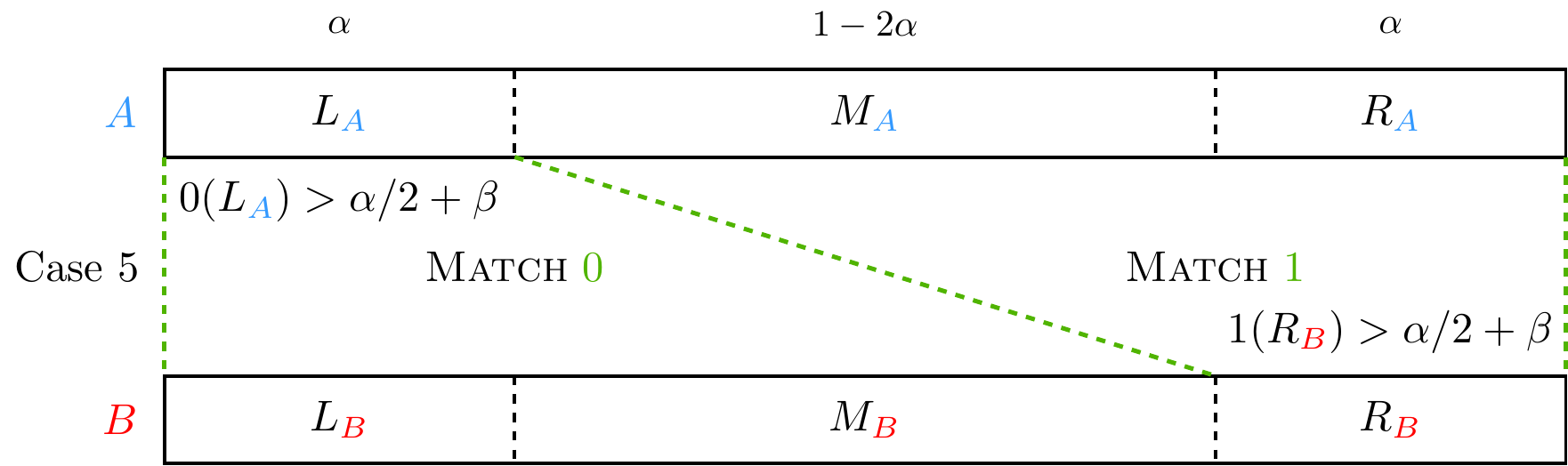}\label{fig:cases_5}}\\
	\vspace{1mm}
	\subfloat[]{\includegraphics[width = 0.9\textwidth]{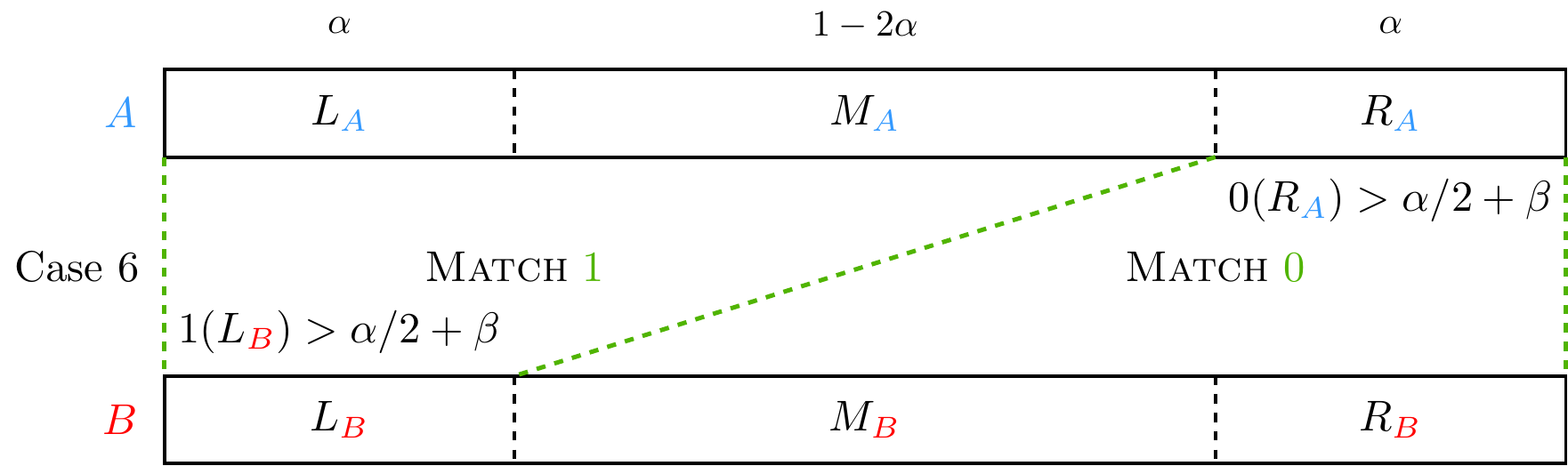}\label{fig:cases_6}}\\
	\caption{Case 3-6.}\label{fig:cases}
\end{figure}

\begin{algorithm}[th]
\begin{algorithmic}[1]\caption{Approximate LCS algorithm}\label{alg:ApproxLCS} 
\Procedure{\textsc{ApproxLCS}}{$A,B,\alpha$}
	\State Split $A$ into three parts, $L_A$, $M_A$ and $R_A$ such that $|L_A| = |R_A| = \alpha$, similarly for $B$
	\State Choose $\beta$ to be sufficiently small constant
	\If{$1(R_B) \leq \alpha/2 + 2\beta$ $0(R_A) \leq \alpha/2 + 2\beta $} \Comment{Case 1}
		\If{$1(R_B) \in [ \alpha / 2 \pm 4\beta]$, $0(R_A) \in [\alpha/2 \pm 4\beta]$} \Comment{Case 1(a)}
			\State $C, \hLB,\hRB \leftarrow \Greedy(L_A \cup M_A, R_A, B)$
      \State $Z \leftarrow \max \{ \min \{ 1 ( L_A \cup M_A ) , 1 (\wh{L_B}) \} , \min \{ 0 (L_A \cup M_A) , 0 ( \wh{L_B} ) \} \}$
			\If{$ Z \leq \alpha / 2 + 10 \beta $}
				\State $C \leftarrow \BMatch(L_A \cup M_A,L_B \cup M_B)$
        \State \hspace{8mm} $ + \max \{ \BMatch(R_A,R_B) , \Approx(R_A,R_B) \}$
			\EndIf
		\ElsIf{$1(R_B) < \alpha / 2 - 4\beta $, $0(R_A) \leq \alpha/2 + 2\beta$} \Comment{Case 1(b)}
			\State $C \leftarrow \Match(A,B,0)$ 
		\ElsIf{$1(R_B) \leq \alpha/2 + 2\beta$, $0(R_A) < \alpha/2- 4\beta$} \Comment{Case 1(c)}
			\State Similar to Case 1(b) 
		\EndIf
	\ElsIf{$1(L_B) \leq \alpha/2 + \beta$ $0(L_A) \leq \alpha/2 + \beta$} \Comment{Case 2}
		\State Similar to Case 1
	\ElsIf{$1(R_B) , 1(L_B) \leq \alpha/2 + \beta$, $0(L_A) , 0(R_A) > \alpha/2 + 2\beta$} \Comment{Case 3}
		\State $C \leftarrow \Match(L_A,L_B,0) + \Match(M_A,M_B,1) + \Match(R_A,R_B,0)$
	\ElsIf{$1(R_B), 1(L_B) > \alpha/2 + 2\beta$, $0(L_A) , 0(R_A) \leq \alpha/2 + \beta$} \Comment{Case 4}
		\State Similar to Case 3
	\ElsIf{$1(R_B) > \alpha/2 + \beta$, $0(L_A) > \alpha/2 + \beta$ } \Comment{Case 5}
		\State $C\leftarrow \Match(L_A, L_B \cup M_B, 0) + \Match(M_A \cup R_A, R_B , 1)$
	\ElsIf{$1(L_B) > \alpha/2 + \beta$, $0(R_A) > \alpha/2 + \beta$} \Comment{Case 6}
		\State $C\leftarrow \Match(L_A \cup M_A, L_B, 1) + \Match(R_A, M_B \cup R_B, 0)$
	\EndIf
	\State \Return $C$
\EndProcedure
\end{algorithmic}
\end{algorithm}